\documentclass[letterpaper,11pt]{article}
\usepackage[utf8]{inputenc}
\usepackage{amsmath}
\usepackage{amsthm, amsfonts, amssymb}
\usepackage{breqn}
\usepackage{graphicx}
\usepackage{mathtools}
\usepackage{appendix}
\usepackage[margin = 1in]{geometry}
\usepackage[section]{placeins}
\usepackage{float}
\usepackage{algorithm, algorithmicx}
\usepackage{comment}
\usepackage{hyperref}
\usepackage{color}
\usepackage{enumitem}
\usepackage{yhmath}
\usepackage{wrapfig}
\usepackage{tabu}

\usepackage[square,sort,comma,numbers]{natbib}
\newtheorem{Definition}{Definition}
\newtheorem{Theorem}{Theorem}
\newtheorem{Lemma}[Theorem]{Lemma}

\newtheorem{fact}{Fact}

\title{Cover Time in Edge-Uniform Stochastically-Evolving Graphs\thanks{This work was supported by the University of Liverpool, EEE/CS School, NeST Initiative}}
\author{
  Ioannis Lamprou%\thanks{Corresponding Author, E-mail: \texttt{Ioannis.Lamprou@liverpool.ac.uk}, Address: Room 310, Ashton Building, Ashton Street, Computer Science Department, University of Liverpool, L69 3BX, UK, Mobile: +447521714871}
  ,
  Russell Martin,
  Paul Spirakis
  \\
  \texttt{\{Ioannis.Lamprou, Russell.Martin, P.Spirakis\}@liverpool.ac.uk}
  \\ \\
  Department of Computer Science,
  \\
  University of Liverpool, UK
  }
  
\begin{document}

\maketitle
\thispagestyle{empty}

\begin{abstract}
	We define a general model of stochastically-evolving graphs, namely the \emph{Edge-Uniform Stochastically-Evolving Graphs}.
	In this model, each possible edge of an underlying general static graph evolves independently being either alive or dead at each discrete time step of evolution following a (Markovian) stochastic rule.
	The stochastic rule is identical for each possible edge and may depend on the past $k \ge 0$ observations of the edge's state.
	We examine two kinds of random walks for a single agent taking place in such a dynamic graph:
	(i) The \emph{Random Walk with a Delay} (\emph{RWD}), where at each step the agent chooses (uniformly at random) an incident possible edge, i.e., an incident edge in the underlying static graph, and then it waits till the edge becomes alive to traverse it.
	(ii) The more natural \emph{Random Walk on what is Available} (\emph{RWA}) where the agent only looks at alive incident edges at each time step and traverses one of them uniformly at random.	
	Our study is on bounding the \emph{cover time}, i.e., the expected time until each node is visited at least once by the agent. 
	For \emph{RWD}, we provide a first upper bound for the cases $k = 0, 1$ by correlating \emph{RWD} with a simple random walk on a static graph.
	Moreover, we present a modified electrical network theory capturing the $k = 0$ case.
	For \emph{RWA}, we derive some first bounds for the case $k = 0$,
	by reducing \emph{RWA} to an \emph{RWD}-equivalent walk with a modified delay.
	Further, we also provide a framework, which is shown to compute the exact value of the cover time for a general family of stochastically-evolving graphs in exponential time.
	Finally, we conduct experiments on the cover time of \emph{RWA} in Edge-Uniform graphs and compare the experimental findings with our theoretical bounds.
\end{abstract}

%\newpage

\section{Introduction}
In the modern era of Internet, modifications in a network topology can occur extremely frequently and in a disorderly way.
Communication links may fail from time to time, while connections amongst terminals may appear or disappear intermittently.
Thus, classical (static) network theory fails to capture such ever-changing processes.
In an attempt to fill this void, different research communities have given rise to a variety of theories on \emph{dynamic networks}.
In the context of algorithms and distributed computing, such networks are usually referred to as \emph{temporal graphs} \cite{Michail}.
A temporal graph is represented by a (possibly infinite) sequence of subgraphs of the same static graph.
That is, the graph is \emph{evolving} over a series of (discrete) time steps under a set of deterministic or stochastic rules of evolution.
Such a rule can be edge- or graph-specific and may take as input graph instances observed in previous time steps.% of the sequence.

In this paper, we focus on stochastically-evolving temporal graphs.
We define a model of evolution, where there exists a single stochastic rule, which is applied \emph{independently} to each edge.
Furthermore, our model is general in the sense that the underlying static graph is allowed to be a general connected graph, i.e.,\ with no further constraints on its topology, and the stochastic rule can include any finite number of past observations. 

Assume now that a single mobile agent is placed on an arbitrary node of a temporal graph evolving under the aforementioned model.
Next, the agent performs a simple random walk; at each time step, after the graph instance is fixed according to the model, the agent chooses uniformly at random a node amongst the neighbors of its current node and visits it. The \emph{cover time} of such a walk is defined as the expected number of time steps until the agent has visited each node at least once.
Herein, we prove some first bounds on the cover time for a simple random walk as defined above, mostly via the use of Markovian theory.%electrical network theory and

Random walks constitute a very important primitive in terms of distributed computing. Examples include their use in information dissemination~\cite{Aleliunas} and random network structure~\cite{Bar-Ilan}; also, see the short survey in~\cite{Bui}.
In this work, we consider a single random walk as a fundamental building block for other more distributed scenarios to follow.

\subsection{Related Work}
A paper very relevant to ours is the one of Clementi, Macci, Monti, Pasquale and Silvestri~\cite{Clementi}, where they consider the flooding time in \emph{Edge-Markovian} dynamic graphs.
In such graphs, each edge independently follows a one-step Markovian rule and their model appears as a special case of ours (matches our case $k=1$).
Further work under this Edge-Markovian paradigm includes \cite{Baumann, Clementi2}. 

Another work related to our paper is the one of Avin, Kouck\'y and Lotker~\cite{Avin}, who define the notion of a \emph{Markovian Evolving Graph}, i.e.,\ a temporal graph evolving over a set of graphs $G_1, G_2,\ldots, $ where the process transits from $G_i$ to $G_j$ with probability $p_{ij}$, and consider random walk cover times. 
Note that their approach becomes computationally intractable if applied to our case; each of the possible edges evolves independently, thence causing the state space to be of size $2^m$, where $m$ is the number of possible edges in our model. 

%Altogether, both above approaches consider a Markovian history of length one.
Clementi, Monti, Pasquale and Silvestri \cite{Clementi3} study the broadcast problem, when at each time step the graph is selected according to the well-known $G_{n,p}$ model. 
Also, Yamauchi, Izumi and Kamei~\cite{Yamauchi} study the rendezvous problem for two agents on a ring, when each edge of the ring independently appears at every time step with some fixed probability $p$.

Moving to a more general scope, research in temporal networks is of interdisciplinary interest, since they are able to capture a wide variety of systems in physics, biology, social interactions and technology.
For a view of the big picture, see the review in \cite{Holme}.
There exist several papers considering, mostly continuous-time, random walks on different models of temporal networks:
In \cite{Starnini}, they consider a walker navigating randomly on some specific empirical networks.
Rocha and Masuda \cite{Rocha} study a lazy version of a random walk, where the walker remains to its current node according to some sojourn probability. 
In \cite{Figueiredo}, they study the behavior of a continuous time random walk on a stationary and ergodic time varying dynamic graph.
Lastly, random walks with arbitrary waiting times are studied in \cite{Delvenne}, while random walks on stochastic temporal networks are surveyed in \cite{Hoffmann}.

In the analysis to follow, we employ several seminal results around the theory of random walks and Markov chains.
For random walks, we base our analysis on the seminal work in \cite{Aleliunas} and the electrical network theory presented in \cite{Chandra, Doyle}. For results on Markov chains, we cite textbooks \cite{Habib, Norris}.

\subsection{Our Results}
We define a general model of stochastically-evolving graphs, where each possible edge evolves independently, but all of them evolve following the same stochastic rule.
Furthermore, the stochastic rule may take into account the last $k$ states of a given edge.
The motivation for such a model lies in several practical examples from networking where the existence of an edge in the recent past means it is likely to exist in the near future, e.g., for telephone or Internet links. 
In some other cases, existence may mean that an edge has "served its purpose" and is now unlikely to appear in the near future, e.g., due to a high maintenance cost. 
The model is a discrete-time one following previous work in the computer science literature.
Moreover, as a first start and for mathematical convenience, it is formalized as a synchronous system, where all possible edges evolve concurrently in distinct rounds (each round corresponding to a discrete time step).

Special cases of our model have appeared in previous literature, e.g., in \cite{Clementi3, Yamauchi} for $k=0$ and in the line of work starting from \cite{Clementi} for $k=1$, however they only consider special graph topologies (like ring and clique).
On the other hand, the model we define is general in the sense that no assumptions, aside from connectivity, are made on the topology of the underlying graph and any amount of history is allowed into the stochastic rule. 
Thence, we believe it can be valued as a basis for more general results to follow capturing search or communication tasks in such dynamic graphs.

We hereby provide the first known bounds relative to the cover time of a simple random walk taking place in such stochastically evolving graphs for $k = 0$.
To do so, we make use of a simple, yet fairly useful, modified random walk, namely the \emph{Random Walk with a Delay} (\emph{RWD}), where at each time step the agent is choosing uniformly at random from the incident edges of the static underlying graph and then waits for the chosen edge to become alive in order to traverse it.
Despite the fact that this strategy may not sound naturally-motivated enough, it can act as a handy tool when studying other, more natural, random walk models as in the case of this paper.
Indeed, we study the natural random walk on such graphs, namely the \emph{Random Walk on What's Available} (\emph{RWA}), where at each time step the agent only considers the currently alive incident edges and chooses to traverse one out of them uniformly at random.

For the case $k = 0$, that is, when each edge appears at each round with a fixed probability $p$ regardless of history, we prove that the cover time for \emph{RWD} is upper bounded by $C_G/p$, where $C_G$ is the cover time of a simple random walk on the (static) underlying graph $G$.
The result can be obtained both by a careful mapping of the \emph{RWD} walk to its corresponding simple random walk on the static graph and by generalizing the standard electrical network theory literature in \cite{Chandra, Doyle}. 
Later, we proceed to prove that the cover time for \emph{RWA} is between $C_G/(1-(1-p)^\Delta)$ and $C_G/(1-(1-p)^\delta)$ where $\delta$, respectively $\Delta$, is the minimum, respectively maximum, degree of the underlying graph.
The main idea here is to reduce \emph{RWA} to an \emph{RWD} walk, where at each step the traversal delay is lower, respectively upper, bounded by $(1-(1-p)^\delta)$, respectively $(1-(1-p)^\Delta)$.

For $k=1$, the stochastic rule takes into account the previous, one time step ago, state of the edge. If an edge was not present, then it becomes alive with probability $p$, whereas if it was alive, then it dies with probability $q$.
For \emph{RWD}, we show a $C_G/\xi_{min}$ upper bound by considering the minimum probability guarantee of existence at each round, i.e., $\xi_{min} = \min\{p, 1-q\}$.
Similarly, we show a $C_G/\xi_{max}$ lower bound, where $\xi_{max} = \max\{p, 1-q\}$.

Consequently, we demonstrate an exact, exponential-time approach to determine the precise cover time value for a general setting of stochastically-evolving graphs, including also the edge-independent model considered in this paper.

Finally, we conduct a series of experiments on calculating the cover time of RWA ($k = 0$ case) on various underlying graphs. We compare our experimental results with the achieved theoretical bounds.

\subsection{Outline}
In Section~\ref{sec:pre}, we provide preliminary definitions and results regarding important concepts and tools that we use in later sections. Then, in Section~\ref{sec:model}, we define our model of stochastically-evolving graphs in a more rigorous fashion. Afterwards, in Sections~\ref{sec:history0} and \ref{sec:history1}, we provide the analysis of our cover time bounds when for determining the current state of an edge we take into account its last $0$ and $1$ states, respectively. 
In Section~\ref{sec:exact}, we demonstrate an exact approach for determining the cover time for general stochastically-evolving graphs.
Then, in Section~\ref{sec:exp}, we present some experimental results on, zero-step history, \emph{RWA} cover time and compare them to the corresponding theoretical bounds in Section~\ref{sec:history0}.
Finally, in Section~\ref{sec:con}, we cite some concluding remarks. 

\section{Preliminaries}\label{sec:pre}
Let us hereby define a few standard notions related to a simple random walk performed by a single agent on a simple connected graph $G = (V,E)$.
By $d(v)$, we denote the degree, i.e.,\ the number of neighbors, of a node $v \in V$. 
A simple random walk is a Markov chain where, for $v, u \in V$, we set $p_{vu} = 1/d(v)$, if $(v,u)\in E$, and $p_{vu} = 0$, otherwise.
That is, an agent performing the walk chooses the next node to visit uniformly at random amongst the set of neighbors of its current node.
Given two nodes $v$, $u$, the expected time for a random walk starting from $v$ to arrive at $u$ is called the \emph{hitting time} from $v$ to $u$ and is denoted by $H_{vu}$.
The \emph{cover time} of a random walk is the expected time until the agent has visited each node of the graph at least once.
Let $P$ stand for the stochastic matrix describing the transition probabilities for a random walk (or, in general, a discrete-time Markov chain) where $p_{ij}$ denotes the probability of transition from node $i$ to node $j$, $p_{ij} \ge 0$ for all $i,j$ and $\sum_j p_{ij} = 1$ for all $i$. 
Then, the matrix $P^t$ consists of the transition probabilities to move from one node to another after $t$ time steps and we denote the corresponding entries as $p_{ij}^{(t)}$. 
Asymptotically, $\lim_{t \rightarrow \infty} P^t$ is referred to as the \emph{limiting distribution} of $P$.
A \emph{stationary distribution} for $P$ is a row vector $\pi$ such that $\pi P = \pi$ and $\sum_i \pi_i = 1$.
That is, $\pi$ is not altered after an application of $P$.
If every state can be reached from another in a finite number of steps, i.e.,\ $P$ is \emph{irreducible}, and the transition probabilities do not exhibit periodic behavior with respect to time, i.e.,\ $gcd\{t: p^{(t)}_{ij} > 0\} = 1$, then the stationary distribution is \emph{unique} and it matches the limiting distribution (\emph{Fundamental Theorem of Markov chains}).
The \emph{mixing time} is the expected number of time steps until a Markov chain approaches its stationary distribution.

In order to derive lower bounds for \emph{RWA}, we use the following graph family, commonly known as \emph{lollipop graphs}, capturing the maximum cover time for a simple random walk, e.g. see \cite{Brightwell, Feige}. 

\begin{Definition}
	A lollipop graph $L^k_n$ consists of a clique on $k$ nodes and a path on $n-k$ nodes connected with a cut-edge, i.e., an edge whose deletion makes the graph disconnected.
\end{Definition}

\section{The Edge-Uniform Evolution Model}\label{sec:model}
Let us define a general model of a dynamically evolving graph.
Let $G = (V, E)$ stand for a simple, \emph{connected} graph, from now on referred to as the \emph{underlying graph} of our model.
The number of nodes is given by $n = |V|$, while the number of edges is denoted by $m = |E|$.
For a node $v \in V$, let $N(v) = \{u: (v, u) \in E \}$ stand for the \emph{open neighborhood} of $v$ and $d(v) = |N(v)|$ for the \emph{(static) degree} of $v$. 
Note that we make no assumptions regarding the topology of $G$, besides connectedness.
We refer to the edges of $G$ as the \emph{possible edges} of our model.
We consider evolution over a sequence of discrete time steps (namely $0, 1, 2, \ldots$) and denote by $\mathcal{G} = (G_0, G_1, G_2, \ldots)$ the infinite sequence of graphs $G_t = (V_t, E_t)$, where $V_t = V$ and $E_t \subseteq E$.
That is, $G_t$ is the graph appearing at time step $t$ and each edge $e \in E$ is either \emph{alive} (if $e \in E_t$) or \emph{dead} (if $e \notin E_t$) at time step $t$.

Let $R$ stand for a \emph{stochastic rule} dictating the probability that a given possible edge is alive at any time step.
We apply $R$ at each time step and at each edge \emph{independently} to determine the set of currently alive edges, i.e.,\ the rule is \emph{uniform} with regard to the edges.
In other words, let $e_t$ stand for a random variable where $e_t = 1$, if $e$ is alive at time step $t$, or $e_t = 0$, otherwise.
Then, $R$ determines the value of $\Pr(e_t = 1 | H_t)$ where $H_t$ is also determined by $R$ and denotes the history length, i.e.,\ the values of $e_{t-1}, e_{t-2}, \ldots$, considered when deciding for the existence of an edge at time step $t$. For instance, $H_t = \emptyset$ means no history is taken into account, while $H_t = \{e_{t-1}\}$ means the previous state of $e$ is taken into account when deciding for its current state. 

Overall, the aforementioned \emph{Edge-Uniform Evolution} model (shortly \emph{EUE}) is defined by the parameters $G$, $R$ and some initial input instance $G_0$.  %
In the following sections, we consider some special cases for $R$ and provide some first bounds for the cover time of $G$ under this model.
Each time step of evolution consists of two stages: in the first stage, the graph $G_t$ is fixed for time step $t$ following $R$, while in the second stage, the agent moves to a node in $N_t[v] = \{v\} \cup \{u \in V: (v,u) \in E_t\}$.
Notice that, since $G$ is connected, then the cover time under \emph{EUE} is finite, since $R$ models edge-specific delays.

\section{Cover Time with Zero-Step History}\label{sec:history0}
We hereby analyze the cover time of $G$ under \emph{EUE} in the special case when no history is taken into consideration for computing the probability that a given edge is alive at the current time step. Intuitively, each edge appears with a fixed probability $p$ at every time step independently of the others.
More formally, for all $e \in E$ and time steps $t$, $\Pr(e_t = 1) = p \in [0,1]$.

\subsection{Random Walk with a Delay}

A first approach toward covering $G$ with a single agent is the following:
The agent is randomly walking $G$ as if all edges were present and, when an edge is not present, it just waits for it to appear in a following time step.
More formally, suppose the agent arrives on a node $v \in V$ with (static) degree $d(v)$ at the second stage of time step $t$. 
Then, after the graph is fixed for time step $t+1$, the agent selects a neighbor of $v$, say $u \in N(v)$, uniformly at random, i.e.,\ with probability $\frac{1}{d(v)}$.
If $(v,u) \in E_{t+1}$, then the agent moves to $u$ and repeats the above procedure.
Otherwise, it remains on $v$ until the first time step $t' > t+1$ such that $(v, u) \in E_{t'}$ and then moves to $u$. 
This way, $p$ acts as a \emph{delay} probability, since the agent follows the same random walk it would on a static graph, but with an expected delay of $\frac{1}{p}$ time steps at each node.
Notice that, in order for such a strategy to be feasible, each node must maintain knowledge about its neighbors in the underlying graph; not just the currently alive ones.
From now on, we refer to this strategy for the agent as the \emph{Random Walk with a Delay} (shortly \emph{RWD}).

Now, let us upper bound the cover time of \emph{RWD} by exploiting its strong correlation to a simple random walk on the underlying graph $G$ via Wald's Equation (Theorem~\ref{thm:Wald}).
Below, let $C_G$ stand for the cover time of a simple random walk on the static graph $G$.

\begin{Theorem}[\cite{Wald}]\label{thm:Wald}
	Let $X_1, X_2, \ldots, X_N$ be a sequence of real-valued, independent and identically distributed random variables where $N$ is a nonnegative integer random variable independent of the sequence (in other words, a stopping time for the sequence). If each $X_i$ and $N$ have finite expectations, then it holds
	$$E[X_1+X_2+\ldots+X_N] = E[N]\cdot E[X_1]$$
\end{Theorem}

\begin{Theorem}\label{thm:RWDAleliunas}
	For any connected underlying graph $G$ evolving under the zero-step history \emph{EUE}, the cover time for \emph{RWD} is expectedly $C_G/p$.
\end{Theorem}
\begin{proof}
	Consider a simple random walk, shortly \emph{SRW}, and an \emph{RWD} (under the \emph{EUE} model) taking place on a given connected graph $G$.
	Given that \emph{RWD} decides on the next node to visit uniformly at random based on the underlying graph, that is, in exactly the same way \emph{SRW} does, we use a coupling argument to enforce \emph{RWD} and \emph{SRW} to follow the exact same trajectory, i.e., sequence of visited nodes.
	
	Then, let the trajectory end when each node in $G$ has been visited at least once and denote by $T$ the total number of node transitions made by the agent.
	Such a trajectory under \emph{SRW} will cover all nodes in expectedly $E[T] = C_G$ time steps.
	On the other hand, in the \emph{RWD} case, for each transition we have to take into account the delay experienced until the chosen edge becomes available.
	Let $D_i \ge 1$ be a random variable, where $1 \le i \le T$ stands for the actual delay corresponding to node transition $i$ in the trajectory. 
	Then, the expected number of time steps till the trajectory is realized is given by
	$E[D_1 + \ldots + D_T]$. 
	Since the random variables $D_i$ are independent and identically distributed by the edge-uniformity of our model, $T$ is a stopping time for them and all of them have finite expectations, then by Theorem~\ref{thm:Wald} we get
	$E[D_1 + \ldots + D_T] = E[T]\cdot E[D_1] = C_G \cdot 1/p$.
\end{proof}

For an explicit general bound on \emph{RWD}, it suffices to use $C_G \le 2m(n-1)$ proved in \cite{Aleliunas}.

\paragraph*{A Modified Electrical Network.}
Another way to analyze the above procedure is to make use of a modified version of the standard literature approach of electrical networks and random walks \cite{Chandra, Doyle}.
This point of view gives us expressions for the hitting time between any two nodes of the underlying graph.
That is, we hereby (in Lemmata~\ref{lem:H=phi}, \ref{lem:commute} and Theorem~\ref{thm:coverRWD}) provide a generalization of the results given in \cite{Chandra, Doyle} thus correlating the hitting and commute times of \emph{RWD} to an electrical network analog and reaching a conclusion for the cover time similar to the one of Theorem~\ref{thm:RWDAleliunas}.
%%% EXTRA MOTIVATION HERE??? 

In particular, given the underlying graph $G$, we design an electrical network, $N(G)$, with the same edges as $G$, but where each edge has a resistance of $r = \frac{1}{p}$ ohms.
Let $H_{u,v}$ stand for the hitting time from node $u$ to node $v$ in $G$, i.e.\ the expected number of time steps until the agent reaches $v$ after starting from $u$ and following \emph{RWD}.
Furthermore, let $\phi_{u,v}$ declare the electrical potential difference between nodes $u$ and $v$ in $N(G)$ when, for each $w \in V$, we inject $d(w)$ amperes of current into $w$ and withdraw $2m$ amperes of current from a single node $v$.
We now upper-bound the cover time of $G$ under \emph{RWD} by correlating $H_{u,v}$ to $\phi_{u,v}$.

\begin{Lemma}\label{lem:H=phi}
	For all $u,v \in V$, $H_{u,v} = \phi_{u,v}$ holds.
\end{Lemma}

\begin{proof}%[Proof of Lemma~\ref{lem:H=phi}]
	Let us denote by $C_{uw}$ the current flowing between two neighboring nodes $u$ and $w$.
	Then, $d(u) = \sum_{w \in N(u)} C_{uw}$ since at each node the total inward current must match the total outward current (Kirchhoff's first law).
	Moving forward, $C_{uw} = \phi_{uw}/r = \phi_{uw}/(1/p) = p\cdot\phi_{uw}$ by Ohm's law.
	Finally, $\phi_{uw} = \phi_{uv} - \phi_{wv}$ since the sum of electrical potential differences forming a path is equal to the total electrical potential difference of the path  (Kirchhoff's second law). 
	Overall, we can rewrite $d(u) = \sum_{w \in N(u)} p (\phi_{u,v} - \phi_{w,v}) = d(u) \cdot p \cdot \phi_{u,v} - p\sum_{w \in N(u)} \phi_{w,v}$.
	Rearranging gives 
	$$ \phi_{u,v} = \frac{1}{p} + \frac{1}{d(u)}\sum_{w \in N(u)} \phi_{w,v} .$$
	
	Regarding the hitting time from $u$ to $v$, we rewrite it based on the first step:
	$$  H_{u,v} = \frac{1}{p} + \frac{1}{d(u)}\sum_{w \in N(u)} H_{w,v} $$
	since the first addend represents the expected number of steps for the selected edge to appear due to \emph{RWD}, and the second addend stands for the expected time for the rest of the walk. 
	
	Wrapping it up, since both formulas above hold for each $u \in V \setminus \{v\}$, therefore inducing two identical linear systems of $n$ equations and $n$ variables, it follows that there exists a unique solution to both of them and $H_{u,v} = \phi_{u,v}$.
\end{proof}

In the lemma below, let $R_{u,v}$ stand for the \emph{effective resistance} between $u$ and $v$, i.e.,\ the electrical potential difference induced when flowing a current of one ampere from $u$ to $v$.

\begin{Lemma}\label{lem:commute}
	For all $u,v \in V$, $H_{u,v} + H_{v,u} = 2mR_{u,v}$ holds.
\end{Lemma}

\begin{proof}%[Proof of Lemma~\ref{lem:commute}]
	Similarly to the definition of $\phi_{u,v}$ above, one can define $\phi_{v,u}$ as the electrical potential difference between $v$ and $u$ when $d(w)$ amperes of current are injected into each node $w$ and $2m$ of them are withdrawn from node $u$. Next, note that changing all currents' signs leads to a new network where for the electrical potential difference, namely $\phi'$, it holds $\phi'_{u,v} = \phi_{v,u}$. 
	We can now apply the Superposition Theorem (see Section 13.3 in \cite{Bird}) and linearly superpose the two networks implied from $\phi_{u,v}$ and $\phi'_{u,v}$ creating a new one where $2m$ amperes are injected into $u$, $2m$ amperes are withdrawn from $v$ and no current is injected or withdrawn at any other node.
	Let $\phi''_{u,v}$ stand for the electrical potential difference between $u$ and $v$ in this last network.
	By the superposition argument, we get $\phi''_{u,v} = \phi_{u,v} + \phi'_{u,v} = \phi_{u,v} + \phi_{v,u}$, while from Ohm's law we get $\phi''_{u,v} = 2m\cdot R_{u,v}$.
	The proof concludes by combining these two observations and applying Lemma~\ref{lem:H=phi}.
\end{proof}

\begin{Theorem}\label{thm:coverRWD}%[Equivalent to Theorem~\ref{thm:RWDAleliunas}]
	For any connected underlying graph $G$ evolving under the zero-step history \emph{EUE}, the cover time for \emph{RWD} is at most $2m(n-1)/p$.
\end{Theorem}

\begin{proof}%[Proof of Theorem~\ref{thm:coverRWD}]
	Consider a spanning tree $T$ of $G$.
	An agent, starting from any node, can visit all nodes by performing an Eulerian tour on the edges of  $T$ (crossing each edge twice).
	This is a feasible way to cover $G$ and thus the expected time for an agent to finish the above task provides an upper bound on the cover time.
	The expected time to cover each edge twice is given by $\sum_{(u,v) \in E_T} (H_{u,v} + H_{v,u})$ where $E_T$ is the edge-set of $T$ with $|E_T| = n-1$.
	By Lemma~\ref{lem:commute}, this is equal to $2m\sum_{(u,v) \in E_T} R_{u,v} = 2m\sum_{(u,v) \in E_T} \frac{1}{p} = 2m(n-1)/p$.
\end{proof}

\subsection{Random Walk on what's Available}
Random Walk with a Delay does provide a nice connection to electrical network theory.
However, depending on $p$, there could be long periods of time where the agent is simply standing still on the same node.
Since the walk is random anyway, waiting for an edge to appear may not sound very wise. 
Hence, we now analyze the strategy of a \emph{Random Walk on what's Available} (shortly \emph{RWA}).
That is, suppose the agent has just arrived at a node $v$ after the second stage at time step $t$ and then $E_{t+1}$ is fixed after the first stage at time step $t+1$.
Now, the agent picks uniformly at random only amongst the alive incident edges at time step $t+1$.
Let $d_{t+1}(v)$ stand for the degree of node $v$ in $G_{t+1}$. 
If $d_{t+1}(v) = 0$, then the agent does not move at time step $t+1$.
Otherwise, if $d_{t+1}(v) > 0$, the agent selects an alive incident edge each having probability $\frac{1}{d_{t+1}(v)}$. The agent then follows the selected edge to complete the second stage of time step $t+1$ and repeats the strategy.
In a nutshell, the agent keeps moving randomly on available edges and only remains on the same node if no edge is alive at the current time step.
Below, let $\delta = \min_{v \in V} d(v)$ and $\Delta = \max_{v \in V} d(v)$.

\begin{Theorem}\label{thm:RWA}
	For any connected underlying graph $G$ with min-degree $\delta$ and max-degree $\Delta$ evolving under the zero-step history \emph{EUE}, the cover time for \emph{RWA} is at least $C_G/(1 - (1-p)^\Delta)$ and at most $C_G/(1 - (1-p)^\delta)$.
\end{Theorem}
\begin{proof}
	Suppose the agent follows \emph{RWA} and has reached node $u \in V$ after time step $t$.
	Then, $G_{t+1}$ becomes fixed and the agent selects uniformly at random a neighboring edge to move to.
	Let $M_{uv}$ (where $v \in \{w \in V: (u,w) \in E\}$) stand for a random variable taking value $1$  if the agent moves to node $v$ and $0$ otherwise.
	For $k = 1, 2, \ldots, d(u) = d$, let $A_k$ stand for the event that $d_{t+1}(u) = k$.
	Therefore, $\Pr(A_k) = \binom{d}{k} p^k(1-p)^{d-k}$ is exactly the probability $k$ out of the $d$ edges exist since each edge exists independently with probability $p$.
	Now, let us consider the probability $\Pr(M_{uv} = 1| A_k)$: the probability $v$ will be reached given that $k$ neighbors are present.
	This is exactly the product of the probability that $v$ is indeed in the chosen $k$-tuple (say $p_1$) and the probability that then $v$ is chosen uniformly at random (say $p_2$) from the $k$-tuple.
	$p_1 = \binom{d-1}{k-1}/\binom{d}{k} = \frac{k}{d}$ since the model is edge-uniform and we can fix $v$ and choose any of the $\binom{d-1}{k-1}$ $k$-tuples with $v$ in them out of the $\binom{d}{k}$ total ones.
	On the other hand, $p_2 = \frac{1}{k}$ by uniformity. Overall, we get $\Pr(M_{uv} = 1| A_k) = p_1 \cdot p_2 = \frac{1}{d}$.
	We can now apply the total probability law to calculate 
	\begin{equation*}
		\resizebox{0.94\linewidth}{!}{
			$	
			\Pr(M_{uv} = 1) = \sum_{k=1}^d \Pr(M_{uv} = 1| A_k)\Pr(A_k) = \frac{1}{d} \sum_{k=1}^d \binom{d}{k} p^k(1-p)^{d-k} = \frac{1}{d}(1 - (1-p)^d)
			$
		}
	\end{equation*}
	To conclude, let us reduce \emph{RWA} to \emph{RWD}. 
	Indeed, in \emph{RWD} the equivalent transition probability is $\Pr(M_{uv} = 1)  = \frac{1}{d}p$, accounting both for the uniform choice and the delay $p$.
	Therefore, the \emph{RWA} probability can be viewed as $\frac{1}{d}p'$ where $p' = (1 - (1-p)^d)$.
	To achieve edge-uniformity we set $p' =  (1 - (1-p)^\delta)$ which lower bounds the delay of each edge and finally we can apply the same \emph{RWD} analysis by substituting $p$ by $p'$. 
	Similarly, we can set the upper-bound delay $p'' = (1 - (1-p)^\Delta)$ to lower-bound the cover time.
	Applying Theorem~\ref{thm:RWDAleliunas} completes the proof.
\end{proof}

The value of $\delta$ used to lower-bound the transition probability may be a harsh estimate for general graphs.
However, it becomes quite more accurate in the special case of a $d$-regular underlying graph where $\delta = \Delta = d$.
To conclude this section, we provide a worst-case lower bound on the cover time based on similar techniques as above.

\begin{Lemma}\label{lem:RWA0-lb}
	There exists an underlying graph $G$ evolving under the zero-step history \emph{EUE} such that the \emph{RWA} cover time is at least $\Omega(mn/(1-(1-p)^\Delta))$.
\end{Lemma}
\begin{proof}
	We consider the $L^{2n/3}_n$ lollipop graph which is known to attain a cover time of $\Omega(mn)$ for a simple random walk \cite{Brightwell, Feige}. 
	Applying the lower bound from Theorem~\ref{thm:RWA} completes the proof.
	\begin{comment}
	Then, we can proceed similarly to the proof of Theorem~\ref{thm:RWA} to calculate the \emph{RWA} delay probability for a given vertex with degree $d$ as $p' = (1-(1-p)^d)$.
	To achieve edge-uniformity, we upper bound $p'$ by $(1-(1-p)^\Delta)$ therefore inducing an expected delay of at least $1/(1-(1-p)^\Delta)$. Applying Theorem~\ref{thm:RWDAleliunas} completes the proof.
	\end{comment}
\end{proof}	

\section{Cover Time with One-Step History}\label{sec:history1}
We now turn our attention to the case where the current state of an edge affects its next state.
That is, we take into account a history of length one when computing the probability of existence for each edge independently.
A Markovian model for this case was introduced in \cite{Clementi}; see Table~\ref{tab:pq}.
The left side of the table accounts for the current state of an edge, while the top for the next one.
The respective table box provides us with the probability of transition from one state to the other.
Intuitively, another way to refer to this model is as the \emph{Birth-Death} model: a dead edge becomes alive with probability $p$, while an alive edge dies with probability $q$.

\begin{table}[H]%\begin{wraptable}{r}{0.25\linewidth}
	\caption{Birth-Death chain for a single edge \cite{Clementi}}
	\label{tab:pq}
	\centering
	\begin{tabular}{|c || c | c|}
		\hline
		& \emph{dead} & \emph{alive} \\
		\hline
		\hline
		\emph{dead}  & $1-p$ & $p$ \\
		\hline
		\emph{alive} & $q$ & $1-q$ \\
		\hline
	\end{tabular}
\end{table}

Let us now consider an underlying graph $G$ evolving under the \emph{EUE} model where each possible edge independently follows the aforementioned stochastic rule of evolution.

\subsection{RWD for General $(p,q)$-Graphs}

Let us hereby derive some first bounds for the cover time of \emph{RWD} via a min-max approach. The idea here is to make use of the "being alive" probabilities to prove lower and upper bounds for the cover time parameterized by $\xi_{min} = \min\{p, 1-q\}$ and $\xi_{max} = \max\{p, 1-q\}$. 

Let us consider an \emph{RWD} walk on a general connected graph $G$ evolving under \emph{EUE} with a zero-step history rule dictating $\Pr(e_t = 1) = \xi_{min}$ for any edge $e$ and time step $t$. We refer to this walk as the \emph{Upper Walk with a Delay}, shortly \emph{UWD}.
Respectively, we consider an \emph{RWD} walk when the stochastic rule of evolution is given by $\Pr(e_t = 1) = \xi_{max}$. We refer to this specific walk as the \emph{Lower Walk with a Delay}, shortly \emph{LWD}. 
Below, we make use of \emph{UWD} and \emph{LWD} in order to bound the cover time of \emph{RWD} in general $(p,q)$-graphs.

\begin{Theorem}\label{thm:RWDxi}
	For any connected underlying graph $G$ and the Birth-Death rule, the cover time of \emph{RWD} is at least $C_G/\xi_{max}$ and at most $C_G/\xi_{min}$.
\end{Theorem}

\begin{proof}
	Regarding \emph{UWD}, one can design a corresponding electrical network where each edge has a resistance of $1/\xi_{min}$ capturing the expected delay till any possible edge becomes alive.
	Applying Theorem~\ref{thm:RWDAleliunas}, gives an upper bound for the \emph{UWD} cover time.
	
	Let $C'$ stand for the \emph{UWD} cover time and $C$ stand for the cover time of \emph{RWD} under the Birth-Death rule.
	It now suffices to show $C \le C'$ to conclude.
	
	In Birth-Death, the expected delay before each edge traversal is either $1/p$, in case the possible edge is dead, or $1/(1-q)$, in case the possible edge is alive. 
	In both cases, the expected delay is upper-bounded by the $1/\xi_{min}$ delay of \emph{UWD} and therefore $C \le C'$ follows since any trajectory under \emph{RWD} will take at most as much time as the same trajectory under \emph{UWD}.
	
	In a similar manner, the cover time of \emph{LWD} lower bounds the cover time of \emph{RWD} and, by applying Theorem~\ref{thm:RWDAleliunas}, we derive a lower bound of $C_G/\xi_{max}$.
\end{proof}

\section{An Exact Approach}\label{sec:exact}
So far, we have established upper and lower bounds for the cover time of edge-uniform stochastically-evolving graphs.
Our bounds are based on combining extended results from simple random walk theory and careful delay estimations.
In this section, we describe an approach to determine the \emph{exact} cover time for temporal graphs evolving under \emph{any} stochastic model.
Then, we apply this approach to the already seen zero-step history and one-step history cases of \emph{RWA}.

The key component of our approach is a Markov chain capturing both phases of evolution: the graph dynamics and the walk trajectory. 
In that case, calculating the cover time reduces to calculating the hitting time to a particular subset of Markov states. 
Although computationally intractable for large graphs, such an approach provides the exact cover time value and is hence practical for smaller graphs.

%\subsection{General Dynamics}
Suppose we are given an underlying graph $G = (V,E)$ and a set of stochastic rules $R$ capturing the evolution dynamics of $G$.
That is, $R$ can be seen as a collection of probabilities of transition from one graph instance to another.
We denote by $k$ the (longest) history length taken into account by the stochastic rules. 
Like before, let $n = |V|$ stand for the number of nodes and $m = |E|$ for the number of possible edges of $G$.
%Moreover, we denote the possible edges of $G$ as $e_1, e_2, \ldots, e_m$.
We define a Markov chain $M$ with states of the form $(H, v, V_c)$, where
\begin{itemize}
	\item $H = (H_1, H_2, \ldots, H_k)$, is a $k$-tuple of \emph{temporal graph instances}, that is, for each $i = 1,2, \ldots, k$, $H_i$ is the graph instance present $i-1$ time steps before the current one (which is $H_1$)
	\item $v \in V(G)$ is the current position of the agent
	\item $V_c \subseteq V(G)$ is the set of already covered nodes, i.e., the set of nodes which have been visited at least once by the agent
\end{itemize}

As described earlier for our edge-uniform model, we assume evolution happens in two phases.
First, the new graph instance is determined according to the rule-set $R$.
Second, the new agent position is determined based on a random walk on what's available.
In this respect, consider a state $S = (H, v, V_c)$ and another state $S' = (H', v', V'_c)$ of the described Markov chain $M$.
Let $\Pr[S \rightarrow S']$ denote the transition probability from $S$ to $S'$.
We seek to express this probability as a product of the probabilities for the two phases of evolution.
The latter is possible, since, in our model, the random walk strategy is independent of the graph evolution.

For the graph dynamics, let $\Pr[H \xrightarrow{R} H']$ stand for the probability to move from a history-tuple $H$ to another history-tuple $H'$ under the rules of evolution in $R$.
Note that, for $i = 1, 2, \ldots, k-1$, it must hold $H'_{i+1} = H_{i}$ in order to properly maintain history, otherwise the probability becomes zero.
On the other hand, for valid transitions, the probability reduces to $\Pr[H'_1|(H_1, H_2, \ldots, H_k)]$, which is exactly the probability that $H'_1$ becomes the new instance given the history $H = (H_1, H_2, \ldots, H_k)$ of past instances (and any such probability is either given directly or implied by $R$).

For the second phase, i.e., the random walk on what's available, we denote by $\Pr[v \xrightarrow{H_j} v']$ the probability of moving from $v$ to $v'$ on some graph instance $H_j$.
Since, the random walk strategy is only based on the current instance, we can derive a general expression for this probability, which is independent of the graph dynamics $R$.
Below, let $N_{H_j}(v)$ stand for the set of neighbors of $v$ in graph instance $H_j$.
If $\{v, v'\} \not\in E(G)$, that is, if there is no possible edge between $v$ and $v'$, then for any temporal graph instance $H_j$, it holds $\Pr[v \xrightarrow{H_j} v'] = 0$.
The probability is also zero for all graph instances $H_j$ where the possible edge is not alive, i.e., $\{v, v'\} \not\in E(H_j)$.
In contrast, if $\{v, v'\} \in E(H_j)$, then $\Pr[v \xrightarrow{H_j} v'] = |N_{H_j}(v)|^{-1}$, since the agent chooses a destination uniformly at random out of the currently alive ones.
Finally, if $v = v'$, then the agent remains still, with probability $1$, only if there exist no alive incident edges. We summarize the above facts in the following equation:

\begin{equation}\label{eq:rwa}
	\Pr[v \xrightarrow{H_j} v'] = 
	\left\{
	\begin{array}{l l}
		1 &  \text{, if } N_{H_j}(v) = \emptyset \text{ and } v' = v\\[2pt]
		|N_{H_j}(v)|^{-1} &  \text{, if } v' \in N_{H_j}(v)\\[2pt]
		0 & \text{, otherwise}
	\end{array}
	\right.
\end{equation}

Overall, we combine the two phases in $M$ and introduce the following transition probabilities.
\begin{itemize}
	\item If $|V_c| < n$:
	$$
	\resizebox{\linewidth}{!}{$
		\Pr[(H, v, V_c)\rightarrow(H', v', V'_c)] = 
		\left\{
		\begin{array}{l l}
		\Pr[H \xrightarrow{R} H']\cdot \Pr[v \xrightarrow{H'_1} v'] & \text{, if } v'\in V'_c \text{ and } V'_c = V_c \\[2pt]
		\Pr[H \xrightarrow{R} H']\cdot \Pr[v \xrightarrow{H'_1} v'] & \text{, if } v'\neq v, v'\not\in V'_c \text{ and } V'_c = V_c \cup \{v'\} \\[2pt]
		0 &, \text{ otherwise }
		\end{array}
		\right.
		$}
	$$
	\item If $|V_c| = n$:
	$$
	\Pr[(H, v, V_c)\rightarrow(H', v', V'_c)] =
	\left\{
	\begin{array}{l l}
	1 & \text{, if } H = H', v = v', V_c = V'_c \\[2pt]
	0 & \text{, otherwise}
	\end{array}
	\right.
	$$
\end{itemize}

For $|V_c| < n$, notice that only two cases may have a non-zero probability with respect to the growth of $V_c$.
If the newly visited node $v'$ is already covered, then $V'_c$ must be identical to $V_c$ since no new nodes are covered during this transition.
Further, if a new node $v'$ is not yet covered, then $V'_c$ is updated to include it as well as all the covered nodes in $V_c$.

For $|V_c| = n$, the idea is that once such a state has been reached, and so all nodes are covered, then there is no need for further exploration.
Therefore, such a state can be made \emph{absorbing}.
In this respect, let us denote the set of these states as $\Gamma = \{(H, v, V_c) \in M \;:\; |V_c| = n\}$.

\begin{Definition}
	Let $\emph{ECT}(G, R)$ be the problem of determining the exact value of the cover time for an \emph{RWA} on a graph $G$ stochastically evolving under rule-set $R$.	
\end{Definition}

\begin{Theorem}\label{thm:exact-general}
	Assume all probabilities of the form $\Pr[H \xrightarrow{R} H']$ used in $M$ are exact reals and known a priori. 
	Then, for any underlying graph $G$ and stochastic rule-set $R$, it holds that $\emph{ECT}(G, R) \in \emph{EXPTIME}$.
\end{Theorem}
\begin{proof}
	% counting
	For each temporal graph instance, $H_i$, in the worst case, there exist $2^m$ possibilities, since each of the $m$ possible edges is either alive or dead at a graph instance.
	For the whole history $H$, the number of possibilities becomes $(2^m)^k = 2^{k\cdot m}$ by taking the product of $k$ such terms.
	There are $n$ possibilities for the walker's position $v$.
	Finally, for each $v \in V(G)$, we only allow states such that $v \in V_c$.
	Therefore, since we fix $v$, there are up to $n-1$ nodes to be included or not in $V_c$ leading to a total of $\mathcal{O}(2^{n-1})$ possibilities for $V_c$. 
	Taking everything into account, $M$ has a total of $\mathcal{O}(2^{k\cdot m + n-1}n)$ states.
	
	% solving
	Let $H_{s, \Gamma}$ stand for the hitting time of $\Gamma$ when starting from a state $s \in M$.
	Assuming exact real arithmetic, we can compute all such hitting times by solving the following system (Theorem 1.3.5 \cite{Norris}):
	$$
	\left\{
	\begin{array}{l l}
	H_{s, \Gamma} = 0 & , \forall s \in \Gamma\\[2pt]
	H_{s, \Gamma} = 1 + \sum_{s' \not\in \Gamma} \Pr[s \rightarrow s']\cdot H_{s', \Gamma} & , \forall s \not\in \Gamma
	\end{array}	
	\right.
	$$
	
	Let $C$ stand for the cover time of an \emph{RWA} on $G$ evolving under $R$.
	By definition, the cover time is the expected time till all nodes are covered, regardless of the position of the walker at that time.
	Consider the set $S = \{(H, v, \{v\}) \in M : v \in V(G)\}$ of start positions for the agent as depicted in $M$.
	Then, it follows $C = \max_{s \in S} H_{s, \Gamma}$, where we take the worst-case hitting time to a state in $\Gamma$ over any starting position of the agent.
	In terms of time complexity, computing $C$ requires computing all values $H_{s, \Gamma}$, for every $s \in S$.
	To do so, one must solve the above linear system of size $\mathcal{O}(2^{k\cdot m + n-1}n)$, which can be done in time exponential to input parameters $n, m$ and $k$.
\end{proof}

It's noteworthy to remark that this approach is general in the sense that there are no assumptions on the graph evolution rule-set $R$ besides it being stochastic, i.e., describing the probability of transition from each graph instance to another given some history of length $k$.
In this regard, Theorem~\ref{thm:exact-general} captures both the case of Markovian Evolving Graphs \cite{Avin} and the case of Edge-Uniform Graphs considered in this paper.
We now proceed and show how the aforementioned general approach applies to the zero-step and one-step history cases of Edge-Uniform Graphs.
To do so, we calculate the corresponding graph-dynamics probabilities.
The random walk probabilities are given in Equation~\ref{eq:rwa}.

\paragraph{RWA on Edge-Uniform Graphs (Zero-Step History).}

Based on the general model, we rewrite the transition probabilities for the special case when \emph{RWA} takes place on an edge-uniform graph without taking into account any memory, i.e., the same case as in Section~\ref{sec:history0}. 
%Let $M^0$ and $M_\gamma^0$ stand for the corresponding Markov chain and its contracted form.
Notice that, since past instances are not considered in this case, the history-tuple reduces to a single graph instance $H$.
We rewrite the transition probabilities, for the case $|V_c| < n$, as follows:
$$
\resizebox{\linewidth}{!}{$
	\Pr[(H, v, V_c)\rightarrow(H, v', V'_c)] = 
	\left\{
	\begin{array}{l l}
	\Pr[H'|H] \cdot \Pr[v \xrightarrow{H'} v'] & \text{, if }  v' \in V'_c \text{ and } V'_c = V_c \\[2pt]
	\Pr[H'|H] \cdot \Pr[v \xrightarrow{H'} v'] & \text{, if }  v'\neq v, v' \not\in V'_c \text{ and } V'_c = V_c \cup \{v'\} \\[2pt]
	0 &, \text{ otherwise }
	\end{array}
	\right.
	$}
$$

% explain the probability
Let $\alpha$ stand for the number of edges alive in $H'$.
Since there is no dependence on history and each edge appears independently with probability $p$, we get $\Pr[H'|H] = \Pr[H'] = p^{\alpha} \cdot (1-p)^{m - \alpha}$.

\paragraph{RWA on Edge-Uniform Graphs (One-Step History).}

We hereby rewrite the transition probabilities for a Markov chain capturing an \emph{RWA} taking place on an edge-uniform graph where, at each time step, the current graph instance is taken into account to generate the next one.
This case is related to the results in Section~\ref{sec:history1}.
Due to the history inclusion, the transition probabilities become more involved than those seen for the zero-history case.
Again, we consider the non-absorbing states, where $|V_c| < n$.

$$
\resizebox{\linewidth}{!}{$
	\Pr[((H_1, H_2), v, V_c)\rightarrow((H'_1, H'_2), v', V'_c)] = 
	\left\{
	\begin{array}{l l}
	\Pr[(H_1, H_2) \rightarrow (H'_1, H'_2)] \cdot \Pr[v \xrightarrow{H'_1} v'] & \text{, if }  v' \in V'_c \text{ and } V'_c = V_c \\[2pt]
	\Pr[(H_1, H_2) \rightarrow (H'_1, H'_2)] \cdot \Pr[v \xrightarrow{H'_1} v'] & \text{, if }  v' \not\in V'_c \text{ and } V'_c = V_c \cup \{v'\} \\[2pt]
	0 &, \text{ otherwise }
	\end{array}
	\right.
	$}
$$

If $H'_2 \neq H_1$, i.e., if it does not hold that, for each $e \in G$, $e \in H'_2$ if and only if $e \in H_1$, then $\Pr[(H_1, H_2) \rightarrow (H'_1, H'_2)] = 0$, otherwise the history is not properly maintained.
On the other hand, if $H'_2 = H_1$, then $\Pr[(H_1, H_2) \rightarrow (H'_1, H'_2)] = \Pr[(H_1, H_2) \rightarrow (H'_1, H_1)] = \Pr[H'_1|H_1]$.
To derive an expression for the latter, we need to consider all edge (mis)matches between $H'_1$ and $H_1$, and properly apply the Birth-Death rule (Table~\ref{tab:pq}).
Below, we denote by $D(H) = E(G) \setminus E(H)$ the set of possible edges of $G$, which are dead at instance $H$.
Let $c_{00} = |D(H_1) \cap D(H'_1)|$, $c_{01} = |D(H_1) \cap E(H'_1)|$, $c_{10} = |E(H_1) \cap D(H'_1)|$ and $c_{11} = |E(H_1) \cap E(H'_1)|$.
Each of the $c_{00}$ edges was dead in $H_1$ and remained dead in $H'_1$, with probability $1-p$.
Similarly, each of the $c_{01}$ edges was dead in $H_1$ and became alive in $H'_1$, with probability $p$.
Also, each of the $c_{10}$ edges was alive in $H_1$ and died in $H'_1$, with probability $q$.
Finally, each of the $c_{11}$ edges was alive in $H_1$ and remained alive in $H'_1$, with probability $1-q$.
Overall, due to the edge-independence of the model, we get $\Pr[H'_1|H_1] = (1-p)^{c_{00}} \cdot p ^{c_{01}} \cdot q^{{c_{10}}} \cdot (1-q)^{{c_{11}}}$.

\section{Experimental Results}\label{sec:exp}

In this section, we discuss some experimental results to complement our previously-established theoretical bounds.
We simulate an \emph{RWA} taking place in graphs evolving under the zero-step history model.
We provide an experimental estimation of the value of the cover time for such a walk.
To do so, for each specific graph and value of $p$ considered, we repeat the experiment a large number of times, e.g., at least $1000$ times. 
In the first experiment, we start from a graph instance with no alive edges.
At each step, after the graph evolves, the walker picks uniformly at random an incident alive edge to traverse.
The process continues till all nodes are visited at least once.
Each next experiment commences with the last graph instance of the previous experiment as its first instance.

We construct underlying graphs in the following fashion: given a natural number $n$, we initially construct a path on $n$ nodes, namely $v_1, v_2, \ldots, v_n$.
Afterward, for each two distinct nodes $v_i$ and $v_j$, we add an edge $\{v_i, v_j\}$ with probability equal to a \emph{randomThreshold} parameter.
For instance, $randomThreshold = 0$ means the graph remains a path.
On the other hand, for $randomThreshold = 1$, the graph becomes a clique. 

In Tables~\ref{tab:ran0.85}, \ref{tab:ran0.5} and \ref{tab:ran0.15}, we display the average cover time, rounding it to the nearest natural number, computed in some indicative experiments for $randomThreshold$ equal to $0.85$, $0.5$ and $0.15$ respectively.   
Consequently, we provide estimates for a lower and an upper bound on the temporal cover time.
In this respect, we experimentally compute a value for the cover time of a simple random walk in the underlying graph, i.e., the \emph{static} cover time.
Then, we plug in this value in place of $C_G$ to apply the bounds given in Theorem~\ref{thm:RWA}.
Overall, the temporal cover times computed appear to be within their corresponding lower and upper bounds.
%Notice that, since our bounds correspond to expected values and are based on an experimental estimation of the static cover time, it may be the case sometimes that the experimental temporal cover time is slightly outside one of the two bounds.

\begin{table}
	\caption{Experimental Results for Randomly-Produced Graphs (\emph{randomThreshold} $=$ $0.85$)}
	\label{tab:ran0.85}
	\centering
	\resizebox{0.91\linewidth}{!}{
	\begin{tabular}{| c | c | c  | c|| c | c | c | c |}
		\hline
		Size & $\delta$ &$\Delta$ & $p$ & Static Cover Time & Temporal Cover Time & Lower Bound & Upper Bound
		\\\hline\hline
		10 & 6 & 9 & 0.9 & 28 & 28 & 28 & 28 \\ \hline
		10 & 7 & 9 & 0.5 & 28 & 28 & 28 & 28 \\ \hline
		10 & 7 & 9 & 0.2 & 27 & 31 & 31 & 34 \\ \hline
		10 & 7 & 9 & 0.1 & 29 & 50 & 47 & 61 \\ \hline
		10 & 7 & 9 & 0.05 & 28 & 78 & 76 & 93 \\ \hline
		10 & 7 & 8 & 0.01 & 28 & 356 & 83 & 413 \\ 
		\hline\hline
		100 & 74 & 92 & 0.9 & 535 & 535 & 535 & 535 \\ \hline
		100 & 74 & 91 & 0.05 & 530 & 543 & 535 & 543 \\ \hline
		100 & 76 & 92 & 0.01 & 536 & 912 & 888 & 1003 \\ \hline
		100 & 74 & 92 & 0.005 & 541 & 1476 & 1465 & 1746 \\ 
		\hline\hline
		250 & 197 & 229 & 0.99 & 1551 & 1551 & 1551 & 1551\\ \hline
		250 & 194 & 228 & 0.75 & 1555 & 1555 & 1555 & 1555\\ \hline
		250 & 192  & 225  & 0.01 & 1548 & 1744 & 1728 & 1810\\ \hline
		250 & 201 & 228 & 0.005 & 1538 & 2326 & 2259 & 2423\\ \hline
		250 & 198 & 225 & 0.001 & 1546 & 7948 & 7670 & 8603\\ \hline
	\end{tabular}
	}
\end{table}

\begin{table}
	\caption{Experimental Results for Randomly-Produced Graphs (\emph{randomThreshold} $=$ $0.5$)}
	\label{tab:ran0.5}
	\centering
	\resizebox{0.91\linewidth}{!}{
	\begin{tabular}{| c | c | c  | c|| c | c | c | c |}
		\hline
		Size & $\delta$ &$\Delta$ & $p$ & Static Cover Time & Temporal Cover Time & Lower Bound & Upper Bound \\\hline\hline
		10 & 3 & 6 & 0.9  & 35 & 35 & 35 & 35 \\ \hline
		10 & 3 & 7 & 0.5  & 33 & 35 & 34 & 38  \\ \hline
		10 & 5 & 8 & 0.2  & 28 & 37 & 33 & 41 \\ \hline
		10 & 4 & 8 & 0.1  & 34 & 69 & 60 & 100  \\ \hline
		10 & 3 & 8 & 0.05  & 32 & 118 & 96 & 226  \\ \hline
		10 & 3 & 7 & 0.01  & 33 & 780 & 486 & 1113  \\ 
		\hline\hline
		100 & 39 & 60 & 0.9  & 542 & 542 & 542 & 542 \\ \hline
		100 & 37 & 68 & 0.1  & 561 & 571 & 561 & 572 \\ \hline
		100 & 35 & 63 & 0.05  & 556 & 589 & 579 & 667 \\ \hline
		100 & 38 & 63 & 0.01  & 544 & 1349 & 1160 & 1714 \\ \hline
		100 & 35 & 61 & 0.005  & 549 & 2436 & 2085 & 3413 \\ \hline
		\hline%\hline
		250 & 106 & 144 & 0.9 & 1589 & 1589 & 1589 & 1589\\ \hline
		250 & 105 & 145 & 0.025 & 1581 & 1646 & 1623 & 1700\\ \hline
		250 & 109 & 147 & 0.01 & 1579 & 2150 & 2046 & 2372\\ \hline
		250 &  105& 150  & 0.005 & 1584  & 3324 & 2998 & 3871 \\ \hline
	\end{tabular}
	}
\end{table}

\begin{table}
	\caption{Experimental Results for Randomly-Produced Graphs (\emph{randomThreshold} $=$ $0.15$)}
	\label{tab:ran0.15}
	\centering
	\resizebox{0.91\linewidth}{!}{
	\begin{tabular}{| c | c | c  | c|| c | c | c | c |}
		\hline
		Size & $\delta$ &$\Delta$ & $p$ & Static Cover Time & Temporal Cover Time & Lower Bound & Upper Bound \\\hline\hline
		10 & 2 & 5 & 0.9  & 38 & 38 & 38 & 38 \\ \hline
		10 & 1 & 5 & 0.5  & 62 & 70 & 64 & 125  \\ \hline
		10 & 2 & 4 & 0.2  & 41 & 88 & 69 & 113 \\ \hline
		10 & 2 & 5 & 0.1  & 48 & 176 & 117 & 252  \\ \hline
		10 & 1 & 5 & 0.05  & 46 & 361 & 203 & 919  \\ \hline
		10 & 2 & 4 & 0.01  & 38 & 1356 & 959 & 1899  \\ 
		\hline\hline
		100 & 9 & 28 & 0.9  & 671 & 671 & 671 & 671 \\ \hline
		100 & 8 & 24 & 0.1  & 634 & 740 & 689 & 1113 \\ \hline
		100 & 11 & 25 & 0.05  & 616 & 1033 & 852 & 1428 \\ \hline
		100 & 9 & 24 & 0.01  & 694 & 4152 & 3240 & 8028 \\ \hline
		100 & 10 & 23 & 0.005  & 642 & 7873 & 5894 & 13127 \\ \hline
		\hline%\hline
		250 & 25 & 57 & 0.9 & 1708 & 1708 & 1708 & 1708\\ \hline
		250 & 27 & 59 & 0.1 & 1700 & 1739 & 1700 & 1803\\ \hline
		250 & 23 & 54 & 0.01 & 1750 & 5167 & 4179 & 8480\\ \hline
		250 & 23 & 54 & 0.005 & 1736 & 9601 & 7321 & 15944\\ \hline
	\end{tabular}
	}
\end{table}

\section{Conclusions}\label{sec:con}
% summary
We defined the general \emph{Edge-Uniform Evolution} model for a stochastically-evolving graph, where a single stochastic rule is applied, but to each edge independently, and provided lower and upper bounds for the cover time of two random walks taking place on such a graph (cases $k = 0, 1$).
Moreover, we provided a general framework to compute the exact cover time of a broad family of stochastically-evolving graphs in exponential time.

% open questions
An immediate open question is to obtain a good lower/upper bound for the cover time of \emph{RWA} in the Birth-Death model.
In this case, the problem becomes quite more complex than the $k = 0$ case.
Depending on the values of $p$ and $q$, the walk may be heavily biased, positively or negatively, toward possible edges incident to the walker's position, which were used in the recent past.

\paragraph{Acknowledgments.}
We would like to acknowledge two anonymous reviewers for spotting technical errors in the previously attempted analysis of the one-step history \emph{RWA}.
Also, we acknowledge another anonymous reviewer, who suggested using Theorem~\ref{thm:RWDAleliunas} as an alternative to electrical network theory and some other useful modifications.

\end{document}